\documentclass[
  aps,
  pra,
  reprint,
  groupedaddress,
  amsmath,
  amssymb,
  floatfix
]{revtex4-2}

\usepackage[utf8]{inputenc}
\usepackage{graphicx}
\usepackage{dcolumn}
\usepackage{bm}
\usepackage{url}
\usepackage{amsthm}

\graphicspath{{Figures/}}

\newtheorem{theorem}{Theorem}
\newtheorem{proposition}{Proposition}
\newtheorem{corollary}{Corollary}
\newtheorem{lemma}{Lemma}
\theoremstyle{definition}
\newtheorem{example}{Example}
\theoremstyle{plain}

\newcommand{\Fone}{F_1}
\newcommand{\Tr}{\operatorname{Tr}}
\newcommand{\RMS}{\operatorname{RMS}}
\newcommand{\BMN}{B_{M,N}}
\newcommand{\GammaOne}{\Gamma}

\newcommand{\Wick}{W}
\newcommand{\abs}[1]{\left\lvert #1 \right\rvert}
\newcommand{\norm}[1]{\left\lVert #1 \right\rVert}

\begin{document}

\title{Maximally antiflat fixed-particle-number fermionic states}

\author{Rahul Gupta}
\author{Arya Haghighian}
\author{James D. Whitfield}
\email{James.D.Whitfield@dartmouth.edu}
\affiliation{%
  Department of Physics and Astronomy,
  Dartmouth College,
  Hanover, New Hampshire 03755, USA
}

\date{\today}

\begin{abstract}
Fermionic Gaussian states are completely characterized by their two-point
correlations, making the covariance matrix a natural reference for
quantifying non-Gaussian structure. Motivated by fermionic antiflatness, we
study its extremal form for pure number-conserving states and ask which
states can be maximally antiflat and how they can be constructed. We show
that, at fixed particle number, maximal first antiflatness is equivalent to
a completely flat one-body reduced density matrix, and we obtain a complete
existence classification together with explicit block and cyclic-orbit
constructions of all allowed sectors. We then investigate the structure that appears beyond 
the covariance level. Maximally antiflat states with
identical one-body data can exhibit distinct connected four-point
correlations, which we characterize through the four-point Majorana Wick
error and, equivalently, the connected two-body reduced density matrix
cumulant. We show that this quartic structure is necessarily nonzero and
formulate its extremization as a two-body reduced density matrix (2RDM) purity
problem. Combining this framework with a Hilbert--Schmidt estimate of Christiansen for truncated two-body correlations, we obtain an upper bound on the connected quartic structure in terms of first antiflatness. Our results establish both a complete characterization of
maximal fermionic antiflatness at fixed particle number and a constrained
hierarchy linking covariance-level flatness to higher-order non-Gaussian
correlations.
\end{abstract}

\maketitle

\section{Introduction}

Fermionic Gaussian states provide one of the fundamental classically
tractable sectors of quantum many-body physics. They arise as ground states
and dynamical states of quadratic fermionic Hamiltonians and, equivalently,
as states generated by fermionic Gaussian or matchgate circuits. Their
special role follows from Wick's theorem: all higher-order correlation
functions are determined by two-point correlations, so a covariance matrix
of polynomial size contains the correlation data required to describe the
state efficiently \cite{SuraceTagliacozzo2022}. Interactions and other
non-Gaussian operations break this closure and generate correlations that
cannot be reconstructed from the covariance matrix alone. Quantifying this
departure from the Gaussian manifold, often called fermionic
non-Gaussianity or fermionic magic, has therefore become a useful way of
characterizing both many-body complexity and resources beyond classically
simulable free-fermion dynamics
\cite{CoffmanSmithGao2025,SierantStornatiTurkeshi2026,HaugTurkeshiSierant2026}.

Entropy provides a natural way to characterize spectral flatness, since
R\'enyi entropies coincide for a uniform distribution and separate as the
spectrum becomes increasingly nonuniform. This idea underlies the stabilizer
R\'enyi entropy, where the relevant distribution is constructed from Pauli
expectation values \cite{LeoneOlivieroHamma2022}, and has also motivated
antiflatness measures of entanglement spectra and their relation to
nonstabilizerness \cite{TirritoEtAl2024}. In the fermionic setting, Sierant, Stornati, and
Turkeshi introduced \emph{fermionic antiflatness} (FAF), a family of
Gaussian-invariant quantities constructed from moments of the Majorana
covariance matrix \cite{SierantStornatiTurkeshi2026}. For a system of
\(M\) fermionic modes, the first antiflatness is defined as
\begin{equation}
F_1
=
M-\frac{1}{2}\operatorname{Tr}(C^{T}C),
\label{eq:intro_f1}
\end{equation}
where \(C\) is the corresponding \(2M\times2M\) Majorana covariance
matrix. Thus, \(F_1\) measures the reduction of the covariance-matrix norm
from its pure-Gaussian value. The covariance matrix also defines a natural
Gaussian reference: the Gaussian state with the same covariance matrix is
the closest Gaussian state in relative entropy \cite{CoffmanSmithGao2025}.
This gives an independent information-theoretic motivation for using
covariance-based quantities to characterize departures from Gaussianity.
FAF vanishes on pure fermionic Gaussian states and provides an efficiently
accessible covariance-level diagnostic of non-Gaussianity. It has since
been applied to many-body phases and dynamics and developed further as an
experimentally accessible witness of fermionic non-Gaussianity
\cite{SierantStornatiTurkeshi2026,HaugTurkeshiSierant2026}.

The present work asks an extremal question about this covariance-level
diagnostic. We restrict to pure states of exactly \(N\) fermions distributed
among \(M\) single-particle modes, where \(N\) denotes the conserved
particle number, and ask: \emph{how large can the first fermionic
antiflatness become once particle number is fixed, and which pure states
attain this value?} The fixed-number setting is both physically common---in
particular in electronic-structure and number-conserving lattice
problems---and mathematically more constrained than the unrestricted
fermionic Hilbert space.
Particle-number conservation eliminates the anomalous quadratic correlators
$\langle a_i a_j\rangle$ and
$\langle a_i^\dagger a_j^\dagger\rangle$, so the remaining quadratic
correlation data are completely determined by the number-conserving
one-body reduced density matrix (1RDM),
$\Gamma_{ij}=\langle a_i^\dagger a_j\rangle$. Its eigenvalues, the natural
occupation numbers, are subject not only to the Pauli bounds but also, for
pure states, to the stronger constraints of the fermionic
\(N\)-representability problem
\cite{Coleman1963,BorlandDennis1972,Klyachko2006,Schilling18}. The
maximal-antiflatness problem therefore becomes a concrete question in
pure-state 1RDM representability.

In the number-conserving sector, the first antiflatness reduces to a simple 1RDM invariant,
\begin{equation}
    F_1=4\left[N-\operatorname{Tr}(\Gamma^2)\right]
=4\operatorname{Tr}\!\left[\Gamma(1-\Gamma)\right],
\end{equation}
a form also noted recently in Ref.~\cite{BonillaVargas2026}. This relation shows that, at fixed particle number, maximizing the first antiflatness is equivalent to making the 1RDM completely uniform, \(\Gamma=(N/M)I_M\). 
The saturation condition therefore places the antiflatness problem within a
broader class of questions concerning pure fermionic states with maximally
mixed low-order reduced density matrices and one-body entanglement
\cite{GigenaRossignoli2015}. Such states have recently been
studied from the perspective of many-body fermionic entanglement and
combinatorial designs \cite{GiorgadzeEtAl2025}. In particular, a maximally
mixed 1RDM is precisely the one-body instance of this broader reduced density matrix problem. Here we connect that condition specifically to maximal
fermionic antiflatness and solve its pure-state realizability problem
analytically for arbitrary \((M,N)\), obtaining a complete existence
classification together with explicit saturating constructions.
Using particle--hole symmetry together with explicit constructions and
obstruction arguments, we identify precisely which particle-number sectors
admit maximally antiflat pure states. We also provide two complementary
constructions. A block construction produces simple equal-weight
superpositions whenever the modes can be partitioned appropriately, while
a cyclic-orbit construction supplies saturating states in the remaining
generic sectors. The exceptional one- and two-particle sectors are
controlled by rank and fermionic pairing constraints, respectively
\cite{SchliemannEtAl2001}.

Maximizing the first antiflatness, however, does not exhaust the
non-Gaussian structure of a state. Once the 1RDM is fixed to its maximally
flat value, many inequivalent pure states remain possible. They are
indistinguishable to all one-body observables and have the same first
antiflatness, yet they can possess different connected two-body
correlations. This raises the second question addressed here:
\emph{what information remains beyond covariance once first antiflatness
has been maximized?} To probe this layer we use the violation of
four-point Wick factorization, an established diagnostic of fermionic
non-Gaussianity \cite{CoffmanSmithGao2025}. In a fixed-particle-number
sector we show that the connected four-Majorana Wick-error tensor and the
connected two-body reduced density matrix cumulant \(\Lambda\) contain the
same quartic information, expressed in different fermionic bases. Writing
\begin{equation}
Q_4=\|\Lambda\|_F,
\end{equation}
we use a standard cumulant identity to show that
\begin{equation}
\operatorname{Tr}\Lambda=-\frac{F_1}{8},
\end{equation}
so this connected quartic structure cannot vanish for any non-Slater
fixed-particle-number state, and in particular for any nontrivial
maximally antiflat state. We further obtain a lower bound on its magnitude.

In addition, the quartic structure is bounded from above by the one-body antiflatness. Translating the Hilbert--Schmidt estimate for the truncated two-body operator proved by Christiansen \cite{Christiansen2024} into the antisymmetric-pair convention used here gives

\begin{equation}
Q_4^2
\leq
\frac{5NF_1}{16}.
\label{eq:intro-q4-upper}
\end{equation}
Thus the covariance-level quantity \(F_1\) does more than diagnose
one-body non-Gaussianity: it constrains the magnitude of the
connected correlations at the next level of the hierarchy. On the
maximally antiflat manifold, where
\(F_1=4N(M-N)/M\), this becomes
\begin{equation}
Q_4^2
\leq
\frac{5N^2(M-N)}{4M}.
\label{eq:intro-q4-upper-flat}
\end{equation}
The bound need not be sharp under the additional flat-1RDM constraint, but
it provides a general analytic upper envelope for the connected quartic
structure of maximally antiflat states.

The resulting picture is therefore hierarchical. First antiflatness
characterizes what is visible at the covariance, or one-body, level; after
that information has been made maximally featureless, connected four-point
correlations form a distinct and necessarily nonvanishing next layer of
non-Gaussian structure. At the same time, this next layer is not
unconstrained: its magnitude is quantitatively linked back to the
one-body antiflatness through Eq.~\eqref{eq:intro-q4-upper}. This
separation and connection are relevant both conceptually and practically:
they provide exact benchmarks for covariance-based non-Gaussianity
measures, supply explicit families of extremal number-conserving states,
and clarify both what first antiflatness fails to determine and what it can
nevertheless constrain about higher-order fermionic correlations.

The remainder of the paper is organized as follows. Section~II fixes the
fermionic and Majorana conventions, establishes the reduction of covariance
data to the 1RDM in a fixed-particle-number sector, and derives the sharp
first-antiflatness bound together with its saturation condition.
Section~III solves the corresponding pure-state existence problem,
including the block and cyclic-orbit constructions and the exceptional
obstructions. Section~IV then moves beyond covariance, relates the
four-point Majorana Wick error to the connected 2RDM cumulant, establishes
the unavoidable connected quartic structure of maximally antiflat states,
recasts its extremization as a 2RDM-purity problem, and derives an upper bound on its magnitude in terms of first antiflatness. Section~V
summarizes the resulting hierarchy and discusses the remaining
representability and higher-order extremal problems.

\section{Antiflatness at Fixed Particle Number}
\label{sec:fixed_number_antiflatness}

We now ask how large the first fermionic antiflatness can become for a state with a fixed number of particles. The key simplification is that particle-number conservation removes anomalous two-point correlators. As a result, the Majorana covariance matrix, and therefore the first antiflatness, can be written entirely in terms of the one-body reduced density matrix (1RDM). This allows us to derive a sharp upper bound and identify the condition for its saturation.

\subsection{Reduction to the one-body density matrix}

Consider \(M\) spinless fermionic modes with creation and annihilation operators
\(a_i^\dagger\) and \(a_i\), satisfying the canonical anticommutation relations
\cite{SuraceTagliacozzo2022},
\begin{equation}
\{a_i,a_j^\dagger\}=\delta_{ij},
\qquad
\{a_i,a_j\}
=
\{a_i^\dagger,a_j^\dagger\}
=
0.
\label{eq:car}
\end{equation}
The total particle-number operator is
\begin{equation}
\hat{N}
=
\sum_{i=1}^{M} a_i^\dagger a_i .
\label{eq:number_operator}
\end{equation}
Throughout this work, we restrict to pure states with a definite particle number,
\begin{equation}
\hat{N}|\psi\rangle
=
N|\psi\rangle .
\label{eq:fixed_number}
\end{equation}

For each fermionic mode, we introduce two Majorana operators,
\begin{equation}
\gamma_{2i-1}
=
a_i+a_i^\dagger,
\qquad
\gamma_{2i}
=
-i(a_i-a_i^\dagger),
\label{eq:majoranas}
\end{equation}
which satisfy
\begin{equation}
\{\gamma_\mu,\gamma_\nu\}
=
2\delta_{\mu\nu}.
\end{equation}
We use the real antisymmetric covariance matrix
\begin{equation}
C_{\mu\nu}
=
-\frac{i}{2}
\langle\psi|
[\gamma_\mu,\gamma_\nu]
|\psi\rangle .
\label{eq:covariance}
\end{equation}
The first fermionic antiflatness is defined as
\cite{SierantStornatiTurkeshi2026}
\begin{equation}
F_1
=
M-\frac{1}{2}\Tr(C^{T}C).
\label{eq:first_antiflatness}
\end{equation}
For a pure fermionic Gaussian state, the covariance matrix satisfies
\cite{SuraceTagliacozzo2022}
\begin{equation}
C^2=-I_{2M}.
\end{equation}
Since \(C\) is antisymmetric, \(C^T=-C\), and therefore
\begin{equation}
C^T C=-C^2=I_{2M}.
\end{equation}
It follows that
\begin{equation}
F_1
=
M-\frac{1}{2}\Tr(I_{2M})
=
0.
\end{equation}
Thus, pure fermionic Gaussian states define the zero-antiflatness
endpoint of the measure.

The relevant one-body object in the number-conserving sector is the 1RDM,
\begin{equation}
\Gamma_{ij}
=
\langle\psi|
a_i^\dagger a_j
|\psi\rangle .
\label{eq:1rdm}
\end{equation}
The matrix \(\Gamma\) is Hermitian and satisfies
\cite{Lowdin1955,Coleman1963}
\begin{equation}
0\leq \Gamma\leq I_M,
\qquad
\Tr\Gamma=N.
\label{eq:1rdm_properties}
\end{equation}
Its eigenvalues \(n_\alpha\) are the natural occupation numbers, with
\(0\leq n_\alpha\leq 1\) and
\begin{equation}
\sum_{\alpha=1}^{M} n_\alpha=N.
\label{eq:natural_occupations}
\end{equation}

Particle-number conservation now gives the main simplification. Since the operators
\(a_i a_j\) and \(a_i^\dagger a_j^\dagger\) change the particle number by two, their expectation values vanish in a state with definite \(N\):
\begin{equation}
\langle a_i a_j\rangle
=
\langle a_i^\dagger a_j^\dagger\rangle
=
0.
\label{eq:anomalous_vanish}
\end{equation}
Thus, all quadratic expectation values are determined by \(\Gamma\).

Using Eqs.~\eqref{eq:majoranas}--\eqref{eq:anomalous_vanish}, the entries of the Majorana covariance matrix can be written directly in terms of the 1RDM. For each mode,
\begin{equation}
C_{2i-1,2i}
=
1-2\Gamma_{ii},
\label{eq:covariance_diagonal}
\end{equation}
while for \(i\neq j\),
\begin{align}
C_{2i-1,2j-1}
&=
2\,\operatorname{Im}\Gamma_{ij},
\\
C_{2i,2j}
&=
2\,\operatorname{Im}\Gamma_{ij},
\\
C_{2i-1,2j}
&=
-2\,\operatorname{Re}\Gamma_{ij},
\\
C_{2i,2j-1}
&=
2\,\operatorname{Re}\Gamma_{ij}.
\label{eq:covariance_offdiagonal}
\end{align}
The remaining entries follow from the antisymmetry of \(C\).

Equations~\eqref{eq:covariance_diagonal}--\eqref{eq:covariance_offdiagonal} show that, at fixed particle number, the Majorana covariance matrix contains no quadratic information beyond the 1RDM. In particular, the first antiflatness is completely determined by \(\Gamma\).

\subsection{Sharp bound and saturation}

We can now express \(F_1\) directly in terms of the 1RDM. From the component relations above,
\begin{equation}
\Tr(C^{T}C)
=
2\sum_{i=1}^{M}
(1-2\Gamma_{ii})^2
+
16\sum_{i<j}
|\Gamma_{ij}|^2.
\label{eq:covariance_norm_gamma}
\end{equation}
Since \(\Gamma\) is Hermitian,
\begin{equation}
\Tr(\Gamma^2)
=
\sum_{i=1}^{M}\Gamma_{ii}^2
+
2\sum_{i<j}|\Gamma_{ij}|^2.
\label{eq:gamma_purity}
\end{equation}
Using \(\Tr\Gamma=N\), Eq.~\eqref{eq:first_antiflatness} therefore becomes
\begin{equation}
F_1
=
4\left(
N-\Tr(\Gamma^2)
\right).
\label{eq:f1_gamma}
\end{equation}
Equivalently, in terms of the natural occupation numbers,
\begin{equation}
F_1
=
4\sum_{\alpha=1}^{M}
n_\alpha(1-n_\alpha).
\label{eq:f1_natural_occupations}
\end{equation}

Equation~\eqref{eq:f1_natural_occupations} expresses the first
antiflatness entirely in terms of the natural occupation spectrum.
We now determine its fixed-\(N\) upper bound.

\begin{theorem}
For any fixed-\(N\) fermionic state,
\begin{equation}
F_1 \leq B_{M,N},
\qquad
B_{M,N}\equiv\frac{4N(M-N)}{M}.
\label{eq:fixedN_bound}
\end{equation}
More precisely,
\begin{equation}
B_{M,N}-F_1
=
4\left\|
\Gamma-\frac{N}{M}I_M
\right\|_F^2.
\label{eq:gap_identity}
\end{equation}
Hence equality holds if and only if
\begin{equation}
\Gamma=\frac{N}{M}I_M.
\label{eq:flat_1rdm}
\end{equation}
\end{theorem}

We use the Frobenius norm
\begin{equation}
\|A\|_F
=
\sqrt{\operatorname{Tr}(A^\dagger A)},
\end{equation}
which coincides with the Hilbert--Schmidt norm for finite-dimensional operators.

\begin{proof}
Using Eq.~\eqref{eq:f1_gamma},
\begin{align}
B_{M,N}-F_1
&=
4\left(
\Tr(\Gamma^2)-\frac{N^2}{M}
\right).
\end{align}
Since \(\Tr\Gamma=N\),
\begin{align}
\left\|
\Gamma-\frac{N}{M}I_M
\right\|_F^2
&=
\Tr\left[
\left(
\Gamma-\frac{N}{M}I_M
\right)^2
\right]
\nonumber\\
&=
\Tr(\Gamma^2)-\frac{N^2}{M}.
\end{align}
This gives Eq.~\eqref{eq:gap_identity}. The Frobenius norm is
nonnegative, which proves the bound. Equality holds exactly when
\(\Gamma=(N/M)I_M\).
\end{proof}
The saturation condition has a simple interpretation. A maximally antiflat state must have
\begin{equation}
n_1=n_2=\cdots=n_M=\frac{N}{M}.
\label{eq:flat_natural_occupations}
\end{equation}
Because all eigenvalues of \(\Gamma\) are equal, the 1RDM is proportional to the identity. Thus, the condition is stronger than requiring equal occupations in one chosen mode basis. Every normalized one-particle orbital has the same average occupation \(N/M\).

\noindent\textit{Relation to relative-entropy non-Gaussianity:} The flat occupation spectrum also has an independent information-theoretic
interpretation. Coffman, Smith, and Gao showed that the Gaussian state
\(\rho_G\) with the same covariance matrix as a state \(\rho\) is the
closest Gaussian state in relative entropy \cite{CoffmanSmithGao2025}:
\begin{equation}
\min_{\sigma_G}D(\rho\Vert\sigma_G)
=
D(\rho\Vert\rho_G).
\label{eq:closest_gaussian}
\end{equation}
For the pure states considered here,
\begin{equation}
D(\rho\Vert\rho_G)=S(\rho_G).
\label{eq:pure_relative_entropy}
\end{equation}

At fixed particle number, the Gaussian reference state is determined by the
natural occupations \(n_\alpha\), and its entropy is
\begin{equation}
\begin{aligned}
S(\rho_G)
&=
\sum_{\alpha=1}^{M} h(n_\alpha),
\\
h(n)
&=
-n\log n-(1-n)\log(1-n).
\end{aligned}
\label{eq:gaussian_entropy}
\end{equation}
Since \(h(n)\) is concave and
\(\sum_{\alpha}n_\alpha=N\), Jensen's inequality gives
\begin{equation}
S(\rho_G)
\leq
M\,h\!\left(\frac{N}{M}\right),
\label{eq:gaussian_entropy_bound}
\end{equation}
with equality exactly when
\begin{equation}
n_\alpha=\frac{N}{M}
\qquad
\text{for all }\alpha.
\label{eq:relative_entropy_flat}
\end{equation}

Thus the same flat occupation spectrum that maximizes the first
antiflatness also maximizes the relative-entropy distance from the Gaussian
set, whenever such a spectrum is realizable by a pure fixed-\(N\) state.

The bound also shows the effect of fixing the filling. At half filling,
\(N=M/2\), the flat 1RDM gives \(C=0\) and
\begin{equation}
B_{M,M/2}=M.
\end{equation}
Away from half filling, \(C\) cannot vanish because the fixed density itself leaves a nonzero quadratic contribution. The maximally antiflat state is therefore the state whose 1RDM is as uniform as the particle-number constraint allows.

Equation~\eqref{eq:gap_identity} solves the optimization problem at the level of one-body data. It does not, however, guarantee that the flat matrix in Eq.~\eqref{eq:flat_1rdm} can arise from a \emph{pure} \(N\)-fermion state. Pure-state 1RDMs obey additional representability constraints beyond
\(0\leq\Gamma\leq I_M\) and \(\Tr\Gamma=N\)
\cite{Coleman1963,Klyachko2006,Schilling18}.
The remaining question is therefore an existence problem:
for which pairs \((M,N)\) does there exist a pure state satisfying
\(\Gamma=(N/M)I_M\)?
More generally, pure-state one-body spectra are constrained by the
fermionic quantum-marginal problem and the associated generalized Pauli
constraints and moment-polytope structure
\cite{WalterDoranGrossChristandl2013,AltunbulakKlyachko2008}.
The next section answers this question and gives explicit constructions whenever saturation is possible.

\section{Existence and Construction of Maximally Antiflat Pure States}
\label{sec:saturators}

Section~\ref{sec:fixed_number_antiflatness} showed that a pure fixed-\(N\)
state reaches the sharp antiflatness bound if and only if
\begin{equation}
    \GammaOne=\frac{N}{M}I.
\end{equation}
The remaining question is therefore one of pure-state representability:
for which pairs \((M,N)\) can this scalar matrix occur as the 1RDM of a
pure \(N\)-fermion state?

It is useful to organize the answer using particle--hole symmetry. If a
pure \(r\)-particle state has
\(\GammaOne=(r/M)I\), then its particle--hole dual in the
\((M-r)\)-particle sector has
\(\GammaOne=((M-r)/M)I\). We therefore define
\begin{equation}
    r=\min(N,M-N)
\end{equation}
and work in the smaller of the particle and hole sectors.

The complete existence result is the following.

\begin{theorem}[Existence classification]
\label{thm:existence-classification}
Let \(r=\min(N,M-N)\). There exists a pure fixed-\(N\) state with
\begin{equation}
    \GammaOne=\frac{N}{M}I
\end{equation}
if and only if one of the following holds:
\begin{enumerate}
    \item \(r=0\);
    \item \(r=2\) and \(M\) is even;
    \item \(r\geq3\).
\end{enumerate}
Equivalently, the only nontrivial impossible sectors are \(r=1\), and
\(r=2\) with odd \(M\).
\end{theorem}

The proof is constructive. We first identify a simple condition on
occupation-basis determinants that guarantees a flat 1RDM, then give two
families of states that realize the allowed sectors. The remaining cases are
ruled out by pure-state constraints.

Closely related support conditions arise in constructions of fermionic states
with maximally mixed few-body reduced density matrices, where combinatorial
designs and restrictions on overlaps between constituent Slater determinants
provide sufficient conditions for low-order RDM flatness
\cite{GiorgadzeEtAl2025}. For the one-body problem considered here, the
relevant condition takes the particularly simple form below.

For a subset \(S\subset\{1,\ldots,M\}\), let \(\lvert S\rangle\) denote the
occupation-basis determinant with precisely the modes in \(S\) occupied.
An off-diagonal 1RDM element
\(\langle a_i^\dagger a_j\rangle\) can receive a contribution only from two
determinants whose occupied sets differ by replacing \(j\) with \(i\).
Thus, for two distinct \(N\)-element supports \(S\) and \(T\), one-body
coherence can occur only when
\begin{equation}
    \lvert S\cap T\rvert=N-1.
\label{eq:one-hop-condition}
\end{equation}
This suggests a simple construction principle: choose equal-weight
determinant supports so that every mode appears equally often, while no two
distinct supports are related by a single particle hop.

\subsection{Block saturators}
\label{subsec:block-saturators}

A particularly simple realization, closely related to the GHZ-like
disjoint-determinant constructions used for maximally mixed fermionic RDMs
\cite{GiorgadzeEtAl2025}, is an equal superposition over disjoint occupied
blocks.

\begin{theorem}[Block saturators]
\label{thm:block-saturators}
Let \(r=\min(N,M-N)\). Assume \(1<r<M\) and \(r\) divides \(M\). Partition the
\(M\) modes into \(K=M/r\) disjoint blocks
\begin{equation}
    B_1,\ldots,B_K,
    \qquad
    \lvert B_k\rvert=r.
\end{equation}
If \(N=r\), define
\begin{equation}
    \lvert\psi\rangle
    =
    \frac{1}{\sqrt K}
    \sum_{k=1}^K \lvert B_k\rangle .
\end{equation}
If \(N=M-r\), define
\begin{equation}
    \lvert\psi\rangle
    =
    \frac{1}{\sqrt K}
    \sum_{k=1}^K \lvert B_k^c\rangle .
\end{equation}
Then \(\GammaOne=(N/M)I\), so \(\lvert\psi\rangle\) saturates the sharp
fixed-\(N\) antiflatness bound.
\end{theorem}

\begin{proof}
First suppose \(N=r\). Each mode belongs to exactly one block, and each block
appears with probability \(1/K\). Hence every diagonal 1RDM entry is
\begin{equation}
    \GammaOne_{jj}
    =
    \frac{1}{K}
    =
    \frac{r}{M}
    =
    \frac{N}{M}.
\end{equation}
For \(i\neq j\), \(\GammaOne_{ij}\) can be nonzero only if two determinants in
the superposition differ by one particle hop. Distinct block determinants
are disjoint \(r\)-element sets. Since \(r>1\), they differ by more than one
particle hop, and therefore all off-diagonal 1RDM entries vanish. Thus
\(\GammaOne=(r/M)I\).

The case \(N=M-r\) is the particle--hole dual. Each mode is absent from
exactly one complement block and is occupied in the remaining \(K-1\) terms,
so its occupation is
\begin{equation}
    \frac{K-1}{K}
    =
    1-\frac{r}{M}
    =
    \frac{M-r}{M}.
\end{equation}
Again, distinct complement-block determinants differ by more than one
particle hop, so all off-diagonal entries vanish.
\end{proof}

\begin{example}
For \(M=4\) and \(N=2\), the two-block state
\begin{equation}
    \lvert\psi\rangle
    =
    \frac{\lvert1100\rangle+\lvert0011\rangle}{\sqrt2}
\end{equation}
has
\begin{equation}
    \GammaOne=\frac12 I
\end{equation}
and therefore attains the fixed-\(N\) antiflatness bound.
\end{example}

The block construction covers all even two-particle sectors and the
half-filled sectors \(M=2r\). It does not cover every allowed pair. For
example, when \((M,N)=(7,3)\), the modes cannot be partitioned into
three-element blocks, even though a flat-1RDM pure state exists. We next use
cyclic symmetry to treat these nondivisible cases.

\subsection{Cyclic-orbit saturators}
\label{subsec:cyclic-saturators}

Label the modes by \(\mathbb{Z}_M\). For \(S\subset\mathbb{Z}_M\), define
\begin{equation}
    S+t
    =
    \{s+t\pmod M:\ s\in S\}.
\end{equation}
An equal-amplitude superposition over all cyclic translates of \(S\)
automatically gives uniform diagonal occupations. The remaining requirement
is that the orbit does not generate the one-particle-hop coherences described
in Eq.~\eqref{eq:one-hop-condition}.

\begin{theorem}[Cyclic-orbit criterion]
\label{thm:cyclic-criterion}
Let \(S\subset\mathbb{Z}_M\) be an \(N\)-element subset whose cyclic orbit has
length \(M\), and define
\begin{equation}
    \lvert\psi_S\rangle
    =
    \frac{1}{\sqrt M}
    \sum_{t\in\mathbb{Z}_M}\lvert S+t\rangle .
\end{equation}
Assume that no nontrivial translate of \(S\) overlaps \(S\) in \(N-1\) sites:
\begin{equation}
    \lvert S\cap(S+t)\rvert\leq N-2,
    \qquad
    t\not\equiv0\pmod M.
\end{equation}
Then
\begin{equation}
    \GammaOne=\frac{N}{M}I,
\end{equation}
and \(\lvert\psi_S\rangle\) saturates the sharp fixed-\(N\) antiflatness
bound.
\end{theorem}

\begin{proof}
Translation invariance forces all diagonal entries of \(\GammaOne\) to be
equal. Since \(\Tr\GammaOne=N\), each diagonal entry is \(N/M\).

For off-diagonal elements, a nonzero contribution to
\(\langle a_i^\dagger a_j\rangle\) requires two determinants in the orbit to
differ by one particle hop. For fixed-\(N\) determinants, this is equivalent
to their supports overlapping in exactly \(N-1\) sites. The overlap
condition excludes this possibility for every nontrivial translate. Hence
all off-diagonal 1RDM entries vanish.
\end{proof}

The cyclic-orbit criterion gives a general condition. The following explicit
family satisfies it in the remaining generic sectors.

\begin{theorem}[Explicit cyclic family]
\label{thm:explicit-cyclic}
Let
\begin{equation}
    r=\min(N,M-N),
    \qquad
    S_r=\{0,1,\ldots, r-3,r-2,r\}\subset\mathbb{Z}_M .
\end{equation}
If \(r\geq3\) and \(M\geq2r+1\), then the equal-amplitude cyclic-orbit state
built from \(S_r\) is a pure fixed-\(r\) state with
\begin{equation}
    \GammaOne=\frac{r}{M}I.
\end{equation}
Its particle--hole dual is a pure fixed-\((M-r)\) state with
\begin{equation}
    \GammaOne=\frac{M-r}{M}I.
\end{equation}
\end{theorem}

\begin{proof}
The set \(S_r\) has \(r\) elements. Since \(M\geq 2r+1\), the cyclic gap from \(r\) back to \(0\)
has length \(M-r>2\), while all other gaps have length one except
for the single gap of length two between \(r-2\) and \(r\).
Hence this long gap is unique, so no nontrivial translation can leave
\(S_r\) invariant. The cyclic orbit therefore has full length \(M\).

It remains to verify the overlap condition. By symmetry,
\[
|S_r\cap(S_r+t)|
=
|S_r\cap(S_r-t)|,
\]
so it is enough to consider
\[
1\leq t\leq \left\lfloor \frac{M}{2}\right\rfloor .
\]
Because \(M\geq 2r+1\), we have
\(r+\lfloor M/2\rfloor\leq M-1\), so no wraparound occurs
for these shifts.

For \(1\leq t\leq r\), we have
\[
S_r+t
=
\{t,t+1,\ldots,t+r-2,t+r\}.
\]
The overlap can then be counted explicitly:
\[
|S_r\cap(S_r+t)|
=
\begin{cases}
r-2, & t=1,\\[3pt]
r-t, & 2\leq t\leq r-2,\\[3pt]
1, & t=r-1,r.
\end{cases}
\]
In every case,
\[
|S_r\cap(S_r+t)|\leq r-2.
\]

For
\[
r<t\leq \left\lfloor \frac{M}{2}\right\rfloor,
\]
the shifted set lies entirely to the right of \(S_r\) without
wrapping around modulo \(M\), so
\[
S_r\cap(S_r+t)=\varnothing.
\]
Thus every nonzero translation satisfies
\[
|S_r\cap(S_r+t)|\leq r-2,
\]
and the cyclic-orbit criterion applies.
\end{proof}

Representative non-block supports are
\begin{align}
    (M,N)=(7,3):\qquad &S_3=\{0,1,3\},
    \\
    (M,N)=(9,4):\qquad &S_4=\{0,1,2,4\}.
\end{align}

\subsection{Obstructions}
\label{subsec:obstructions}

The two constructions establish existence for all \(r\geq3\) sectors and for
\(r=2\) when \(M\) is even. The remaining cases are \(r=1\) and \(r=2\) with
odd \(M\). These are genuine pure-state obstructions rather than limitations
of the constructions.

\begin{lemma}[One-particle rank obstruction]
\label{lem:rank-obstruction}
If \(M>1\), no pure one-particle state on \(M\) modes has
\begin{equation}
    \GammaOne=\frac{1}{M}I.
\end{equation}
By particle--hole symmetry, no pure one-hole state has
\begin{equation}
    \GammaOne=\frac{M-1}{M}I.
\end{equation}
\end{lemma}

\begin{proof}
A pure one-particle state has the form
\begin{equation}
    \lvert\psi\rangle
    =
    \sum_{j=1}^M c_j a_j^\dagger\lvert0\rangle,
    \qquad
    \sum_j\lvert c_j\rvert^2=1.
\end{equation}
Its 1RDM is
\begin{equation}
    \GammaOne=cc^\dagger,
\end{equation}
which is a rank-one projection. By contrast, \((1/M)I\) has rank \(M\).
The two can therefore agree only in the trivial one-mode case. The one-hole
statement follows by particle--hole symmetry.
\end{proof}

\begin{lemma}[Odd two-fermion obstruction]
\label{lem:two-fermion-obstruction}
If \(M\) is odd, no pure two-particle state on \(M\) modes has
\begin{equation}
    \GammaOne=\frac{2}{M}I.
\end{equation}
By particle--hole symmetry, no pure two-hole state has
\begin{equation}
    \GammaOne=\frac{M-2}{M}I.
\end{equation}
\end{lemma}

\begin{proof}
The canonical form of a pure two-fermion state implies that its nonzero
natural occupation numbers occur in degenerate pairs
\cite{SchliemannEtAl2001}. A completely flat two-particle 1RDM in odd
dimension would instead have \(M\) equal nonzero eigenvalues \(2/M\), which
is incompatible with pair degeneracy. The two-hole statement follows by
particle--hole symmetry.
\end{proof}

\subsection{Completion of the classification}
\label{subsec:classification}

We now complete the proof of
Theorem~\ref{thm:existence-classification}. The cases \(r=0\) correspond to
the vacuum and fully occupied sectors and are immediate. For \(r=2\) and
even \(M\), Theorem~\ref{thm:block-saturators} applies by partitioning the
modes into two-element blocks.

For \(r\geq3\), the half-filled case \(M=2r\) is covered by the block
construction with two \(r\)-element blocks. If \(M\geq2r+1\),
Theorem~\ref{thm:explicit-cyclic} gives an explicit cyclic saturator. Since
\(r\leq M/2\), these cases exhaust all sectors with \(r\geq3\).

The only remaining possibilities are \(r=1\) and \(r=2\) with odd \(M\),
which are ruled out by Lemmas~\ref{lem:rank-obstruction} and
\ref{lem:two-fermion-obstruction}, respectively. This completes the proof of
Theorem~\ref{thm:existence-classification}.

Together with the sharp bound derived in Section~II, the classification gives
the complete answer to the fixed-\(N\) saturation problem. In every allowed
sector,
\begin{equation}
    \Fone=\BMN
\end{equation}
is attained by an explicit flat-1RDM state. In the excluded sectors,
\(\GammaOne=(N/M)I\) cannot be realized by any pure state, and the sharp
bound therefore cannot be saturated.

\begin{figure}[htbp]
    \centering
    \includegraphics[width=0.72\linewidth]{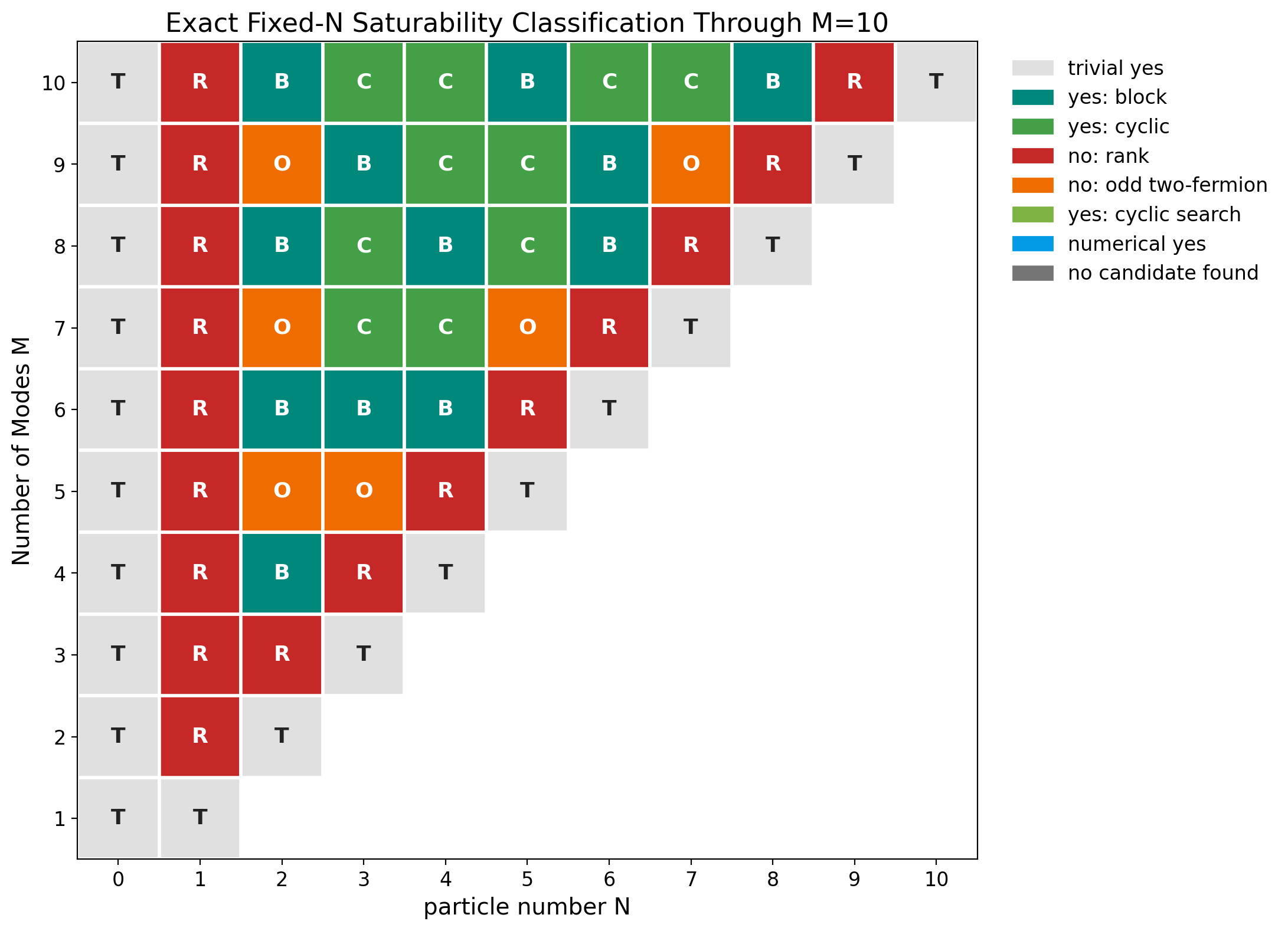}
    \caption{Small-sector saturability classification through \(M\leq10\).
    \(T\): trivial vacuum/full sectors. \(B\): block saturators. \(C\):
    cyclic-family saturators. \(R\): one-particle or one-hole rank obstruction.
    \(O\): odd-\(M\) two-fermion or two-hole obstruction.}
    \label{fig:saturability-heatmap}
\end{figure}

\section{Beyond Maximal First Antiflatness:
Connected Four-Point Structure}
\label{sec:wick-error}

The previous sections identify and classify the pure fixed-\(N\) states that
maximize the first fermionic antiflatness. In every allowed sector these
states satisfy
\begin{equation}
\Gamma=\frac{N}{M}I_M,
\end{equation}
and therefore attain the same maximal value \(F_1=B_{M,N}\). As discussed
in Sec.~II, the same flat occupation spectrum also maximizes the
relative-entropy distance from the Gaussian set at fixed \((M,N)\).

These statements do not imply that all maximizers have the same many-body
correlations. First antiflatness depends only on the covariance matrix, or
equivalently on the 1RDM in a fixed-particle-number sector. Once this
one-body information is fixed, connected correlations at higher order remain
undetermined. We therefore ask a second extremal question: how much
four-point structure can vary among states that are already maximally
antiflat?

\subsection{Four-point Wick error and the connected 2RDM}

For distinct Majorana indices
\(\alpha<\beta<\mu<\nu\), define the four-point Wick-error tensor
\begin{align}
W_{\alpha\beta\mu\nu}
&=
\langle
\gamma_\alpha\gamma_\beta\gamma_\mu\gamma_\nu
\rangle
-
\langle\gamma_\alpha\gamma_\beta\rangle
\langle\gamma_\mu\gamma_\nu\rangle
\nonumber\\
&\quad
+
\langle\gamma_\alpha\gamma_\mu\rangle
\langle\gamma_\beta\gamma_\nu\rangle
-
\langle\gamma_\alpha\gamma_\nu\rangle
\langle\gamma_\beta\gamma_\mu\rangle .
\label{eq:majorana-wick-error}
\end{align}
For a fermionic Gaussian state this tensor vanishes by Wick's theorem
\cite{CoffmanSmithGao2025}. We characterize its overall magnitude by
\begin{align}
  \norm{\Wick}_2^2
  =
  \sum_{\alpha<\beta<\mu<\nu}
  \abs{\Wick_{\alpha\beta\mu\nu}}^2,
  \nonumber\\
  \RMS(\Wick)
  =
  \left(
  \frac{\norm{\Wick}_2^2}{\binom{2M}{4}}
  \right)^{1/2}.
  \label{eq:wick-rms}
\end{align}
The RMS normalization separates the strength of a typical connected
four-Majorana entry from the combinatorial growth in the number of Majorana
quadruples.

In a fixed-\(N\) sector, the same connected quartic information can be
written entirely in number-conserving form. Expanding a four-Majorana
correlator produces quartic monomials containing different numbers of
creation and annihilation operators. Any term with a nonzero net change in
particle number has vanishing expectation value in a state with definite
\(N\). Therefore, only the terms containing two creation and two
annihilation operators contribute. These are precisely the matrix elements
of the two-body reduced density matrix,
\begin{equation}
D_{ij;kl}
=
\langle
a_i^\dagger a_j^\dagger a_l a_k
\rangle .
\label{eq:2rdm}
\end{equation}
Its Wick prediction from the 1RDM is
\begin{equation}
D^{\mathrm{Wick}}_{ij;kl}
=
\Gamma_{ik}\Gamma_{jl}
-
\Gamma_{il}\Gamma_{jk},
\end{equation}
The connected, or cumulant, part of the 2RDM is defined by
\cite{Mazziotti1998,KutzelniggMukherjee1999}
\begin{equation}
\Lambda_{ij;kl}
=
D_{ij;kl}
-
\left(
\Gamma_{ik}\Gamma_{jl}
-
\Gamma_{il}\Gamma_{jk}
\right).
\label{eq:connected-2rdm}
\end{equation}

\begin{proposition}[Fixed-\(N\) quartic equivalence]
\label{prop:wick-connected-2rdm}
For a fixed-particle-number state, the four-Majorana Wick-error tensor
\(W\) and the connected 2RDM cumulant \(\Lambda\) encode the same connected
quartic correlations, written in different fermionic bases.
\end{proposition}

\begin{proof}
The creation and annihilation operators are linear combinations of
Majorana operators,
\begin{equation}
a_j
=
\frac{\gamma_{2j-1}+i\gamma_{2j}}{2},
\qquad
a_j^\dagger
=
\frac{\gamma_{2j-1}-i\gamma_{2j}}{2}.
\end{equation}
The transformation between the Majorana operators and the operators
\(a_j,a_j^\dagger\) is linear and invertible, so the connected
fourth-order correlators in the two representations determine one another.
Expanding a four-Majorana product therefore produces quartic monomials in
\(a\) and \(a^\dagger\). In a fixed-\(N\) state, only the
number-conserving terms with two creation and two annihilation operators
can have nonzero expectation values. Terms generated by fermionic
anticommutation at lower order are precisely removed by the disconnected
Wick contractions. The remaining connected terms are the components of
Eq.~\eqref{eq:connected-2rdm}. Thus \(W\) and \(\Lambda\) contain the same
connected quartic information, expressed in different fermionic bases.
\end{proof}

Thus the 1RDM and connected 2RDM provide the first two levels of the
fixed-\(N\) correlation hierarchy:
\begin{equation}
\Gamma
\quad\longrightarrow\quad
\Lambda
\quad\longrightarrow\quad
\text{higher connected cumulants}.
\end{equation}
The first level fixes \(F_1\), while the second resolves correlations that
remain invisible after the one-body data have been made completely flat.

The distinction between these two levels has a simple interpretation in
the occupation basis. The 1RDM sees diagonal occupations and coherences
between determinants whose supports differ by a single particle hop,
\begin{equation}
|S\cap T|=N-1.
\end{equation}
By contrast, the off-diagonal part of \(\Lambda\) records pair-transfer
coherences between determinants that differ by two particle moves,
\begin{equation}
|S\cap T|=N-2.
\end{equation}
The diagonal part records occupation correlations such as
\begin{equation}
\langle n_i n_j\rangle
-
\langle n_i\rangle\langle n_j\rangle.
\end{equation}
For a flat-1RDM saturator, all
\(\langle n_i\rangle=N/M\) are fixed, so any remaining two-mode
correlation structure must live in \(\Lambda\).

Flat one-body data therefore do not determine \(\Lambda\). Figure~2 shows
the exact four-point Wick-error RMS for the block and cyclic saturators as
a function of \(M\). All states within a given \((M,N)\) sector share the
same flat 1RDM and the same maximal first antiflatness, but the connected
tensor \(\Lambda\) is not fixed by these data. Its Frobenius norm is,
however, rigid in the smallest sectors: for \(N\leq 3\),
\(\|\Lambda\|_F\) takes the same value on the entire flat manifold, and
for \(N=2\) all flat states are related by a single-particle unitary
\cite{SchliemannEtAl2001}. Differences in the magnitude of connected
quartic structure between saturators in the same sector can therefore
first appear at \(N\geq 4\).

\begin{figure}[t]
    \centering
    \includegraphics[width=\columnwidth]
    {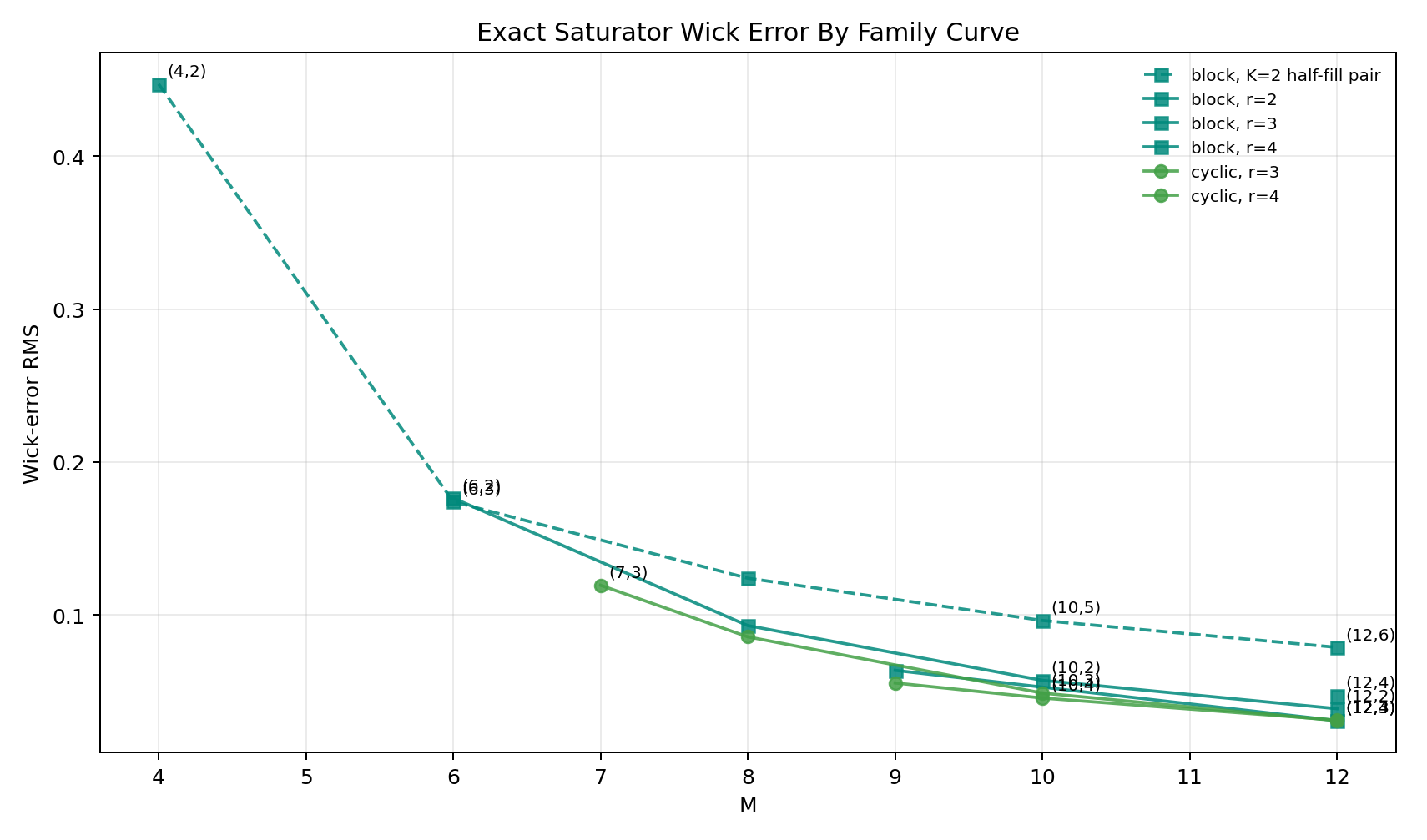}
    \caption{Exact four-point Wick-error RMS for block
and cyclic saturators versus \(M\), grouped by family. All states shown have
flat 1RDMs and maximal first antiflatness; the one-body data do not fix the
connected quartic tensor, although for \(N\leq 3\) its Frobenius norm is
fixed within each sector.}
    \label{fig:wick-family-rms}
\end{figure}

Four-point Wick RMS should not be interpreted as a complete distance from
Gaussianity. A fermionic Gaussian state factorizes at every even order,
while this diagnostic checks only the first connected correction beyond
covariance. Important non-Gaussian structure can be spread across many quadruples,
or can first appear at six-point and higher order. For the fixed-\(N\) hierarchy studied here, however, \(\Lambda\) is the
natural next object after \(\Gamma\). Having seen that flat one-body data
do not determine \(\Lambda\), we now ask whether this connected quartic
structure can ever vanish altogether. As we show next, the answer is no:
maximal one-body flatness at fixed particle number necessarily leaves a
nonzero connected four-point contribution.

\subsection{Four-point correlations are unavoidable}

The connected four-point structure of a maximally antiflat state cannot
vanish completely. To see this, regard \(D\) and \(\Lambda\) as matrices on
the antisymmetric two-particle space
\(\wedge^2\mathbb{C}^M\), indexed by pairs \(i<j\). Its dimension is
\begin{equation}
d_2=\binom{M}{2}.
\end{equation}

For a flat 1RDM,
\begin{equation}
\Gamma=pI_M,
\qquad
p=\frac{N}{M},
\end{equation}
so the Wick contribution on this pair space is simply
\begin{equation}
D^{\mathrm{Wick}}
=
p^2 I_{d_2}.
\end{equation}
Hence
\begin{equation}
\Lambda
=
D-p^2I_{d_2}.
\label{eq:lambda-flat}
\end{equation}

The trace of the connected 2RDM can in fact be related directly to the
first antiflatness for any fixed-\(N\) state.

\begin{proposition}[Unavoidable quartic correlations]
\label{prop:unavoidable-quartic}
For any fixed-\(N\) state,
\begin{equation}
\operatorname{Tr}\Lambda
=
-\frac{F_1}{8}.
\label{eq:lambda-trace-general}
\end{equation}
In particular, for every maximally antiflat pure state with
\(0<N<M\),
\begin{equation}
\operatorname{Tr}\Lambda
=
-\frac{N(M-N)}{2M}.
\label{eq:lambda-trace}
\end{equation}
Hence
\begin{equation}
\Lambda\neq0.
\end{equation}
Therefore every nontrivial maximally antiflat fixed-\(N\) state has a
nonzero connected four-point cumulant.
\end{proposition}

\begin{proof}
The trace of the 2RDM on the antisymmetric pair space is
\begin{equation}
\operatorname{Tr}D
=
\sum_{i<j}
\langle n_i n_j\rangle.
\end{equation}
Using
\begin{equation}
\sum_{i<j}n_i n_j
=
\frac{1}{2}
\left(
\hat N^2-\hat N
\right),
\end{equation}
a state with definite particle number satisfies
\begin{equation}
\operatorname{Tr}D
=
\binom{N}{2}.
\end{equation}

For the Wick contribution,
\begin{align}
\operatorname{Tr}D^{\mathrm{Wick}}
&=
\sum_{i<j}
\left(
\Gamma_{ii}\Gamma_{jj}
-
|\Gamma_{ij}|^2
\right)
\nonumber\\
&=
\frac{1}{2}
\left[
(\operatorname{Tr}\Gamma)^2
-
\operatorname{Tr}(\Gamma^2)
\right].
\end{align}
Since \(\operatorname{Tr}\Gamma=N\),
\begin{align}
\operatorname{Tr}\Lambda
&=
\binom{N}{2}
-
\frac{1}{2}
\left[
N^2-\operatorname{Tr}(\Gamma^2)
\right]
\nonumber\\
&=
\frac{1}{2}
\left[
\operatorname{Tr}(\Gamma^2)-N
\right].
\end{align}
Using
\begin{equation}
F_1
=
4\left(
N-\operatorname{Tr}(\Gamma^2)
\right)
\end{equation}
gives Eq.~\eqref{eq:lambda-trace-general}.

For a maximally antiflat state,
\begin{equation}
F_1
=
B_{M,N}
=
\frac{4N(M-N)}{M},
\end{equation}
and Eq.~\eqref{eq:lambda-trace-general} therefore gives
Eq.~\eqref{eq:lambda-trace}. For \(0<N<M\), this quantity is nonzero.
\end{proof}

The result is stronger than the statement that different maximally
antiflat states \emph{may} differ at fourth order. Connected quartic
correlations are required by fixed particle number itself once the 1RDM is
made completely flat. More generally, Eq.~\eqref{eq:lambda-trace-general}
shows that first antiflatness directly fixes the trace of the connected
quartic cumulant.

A quantitative lower bound follows immediately.

\begin{corollary}[Universal quartic lower bound]
\label{cor:lambda-lower-bound}
Every maximally antiflat state with \(0<N<M\) satisfies
\begin{equation}
\|\Lambda\|_F^2
\geq
\frac{N^2(M-N)^2}
{2M^3(M-1)}.
\label{eq:lambda-lower-bound}
\end{equation}
Equality holds only when the 2RDM is proportional to the identity on
\(\wedge^2\mathbb{C}^M\):
\begin{equation}
D
=
\frac{\binom{N}{2}}
{\binom{M}{2}}
I_{d_2}.
\label{eq:maximally-mixed-2rdm}
\end{equation}
\end{corollary}

\begin{proof}
For any Hermitian \(d_2\times d_2\) matrix,
\begin{equation}
|\operatorname{Tr}\Lambda|^2
\leq
d_2\|\Lambda\|_F^2.
\end{equation}
Substituting Eq.~\eqref{eq:lambda-trace} and
\(d_2=\binom{M}{2}\) gives
Eq.~\eqref{eq:lambda-lower-bound}. Equality in the matrix
Cauchy--Schwarz inequality requires \(\Lambda\) to be proportional to the
identity, which together with Eq.~\eqref{eq:lambda-flat} gives
Eq.~\eqref{eq:maximally-mixed-2rdm}.
\end{proof}

The equality condition corresponds to a maximally mixed normalized 2RDM.
Such higher-body maximally mixed fermionic RDMs have been studied in the
context of many-body fermionic entanglement and combinatorial designs
\cite{GiorgadzeEtAl2025}. Their existence is a stronger requirement than
flatness of the 1RDM and is not guaranteed for every \((M,N)\).

\subsection{A second extremal problem}
\label{subsec:second-extremal}

The first-antiflatness problem determines which pure states can make the
one-body data completely flat. It does not, however, distinguish between
the different many-body states that realize this condition. We therefore
consider the set
\begin{equation}
\mathcal{S}_{M,N}^{\mathrm{flat}}
=
\left\{
|\psi\rangle:
\Gamma_\psi=\frac{N}{M}I_M
\right\}.
\label{eq:flat-manifold}
\end{equation}
All states in this set have the same maximal first antiflatness, while their
connected two-body correlations may differ.

The Frobenius norm of the two-particle cumulant has previously been used
as a measure of fermionic correlation and entanglement
\cite{JuhaszMazziotti2006}. We therefore define
\begin{equation}
Q_4(\psi)
=
\|\Lambda(\psi)\|_F.
\label{eq:q4}
\end{equation}
This leads to a second extremal problem,
\begin{align}
Q_4^{\min}(M,N)
&=
\min_{\psi\in\mathcal{S}_{M,N}^{\mathrm{flat}}}
Q_4(\psi),
\\
Q_4^{\max}(M,N)
&=
\max_{\psi\in\mathcal{S}_{M,N}^{\mathrm{flat}}}
Q_4(\psi).
\label{eq:q4-extrema}
\end{align}
Unlike the first-antiflatness optimization, these extrema are not fixed by
the 1RDM alone.

The problem has a useful reduced density matrix interpretation.

\begin{proposition}[Quartic norm and 2RDM purity]
\label{prop:q4-purity}
For every state in
\(\mathcal{S}_{M,N}^{\mathrm{flat}}\),
\begin{equation}
\|\Lambda\|_F^2
=
\operatorname{Tr}(D^2)
-
2\left(\frac{N}{M}\right)^2
\binom{N}{2}
+
\left(\frac{N}{M}\right)^4
\binom{M}{2}.
\label{eq:lambda-purity}
\end{equation}
Consequently, at fixed \((M,N)\), extremizing \(Q_4\) is equivalent to
extremizing the purity of the normalized two-particle reduced state.
\end{proposition}

\begin{proof}
For a flat 1RDM,
\begin{equation}
\Lambda
=
D-p^2 I_{d_2},
\qquad
p=\frac{N}{M},
\qquad
d_2=\binom{M}{2}.
\end{equation}
Therefore,
\begin{align}
\|\Lambda\|_F^2
&=
\operatorname{Tr}
\left[
(D-p^2I_{d_2})^2
\right]
\nonumber\\
&=
\operatorname{Tr}(D^2)
-
2p^2\operatorname{Tr}D
+
p^4d_2.
\end{align}
Using
\begin{equation}
\operatorname{Tr}D=\binom{N}{2}
\end{equation}
gives Eq.~\eqref{eq:lambda-purity}.
\end{proof}

To interpret this result, define the normalized 2RDM
\begin{equation}
\rho_2
=
\frac{D}{\binom{N}{2}},
\qquad
\operatorname{Tr}\rho_2=1.
\label{eq:normalized-2rdm}
\end{equation}
If
\begin{equation}
\rho_2
=
\sum_a
\lambda_a
|g_a\rangle\langle g_a|,
\end{equation}
then
\begin{equation}
P_2
\equiv
\operatorname{Tr}(\rho_2^2)
=
\sum_a\lambda_a^2,
\label{eq:2rdm-purity}
\end{equation}
and
\begin{equation}
\operatorname{Tr}(D^2)
=
\binom{N}{2}^{\,2}P_2.
\end{equation}

The eigenvectors \(|g_a\rangle\) represent orthogonal two-fermion modes.
A small \(P_2\) means that the two-particle weight is distributed over many
such modes, whereas a large \(P_2\) means that it is concentrated in fewer
two-particle configurations. Thus, within the flat-1RDM manifold, \(Q_4\)
measures how strongly the remaining two-body structure is concentrated once
all one-body information has already been made uniform.

Substituting Eq.~\eqref{eq:2rdm-purity} into
Eq.~\eqref{eq:lambda-purity} gives
\begin{equation}
Q_4^2
=
\binom{N}{2}^{\,2}P_2
-
2\left(\frac{N}{M}\right)^2
\binom{N}{2}
+
\left(\frac{N}{M}\right)^4
\binom{M}{2}.
\label{eq:q4-p2}
\end{equation}
At fixed \((M,N)\), all terms except \(P_2\) are constants. Therefore
\begin{equation}
Q_4^{\min}
\longleftrightarrow
P_2^{\min},
\qquad
Q_4^{\max}
\longleftrightarrow
P_2^{\max}.
\label{eq:q4-purity-extrema}
\end{equation}

The general density-matrix bounds
\begin{equation}
\frac{1}{\binom{M}{2}}
\leq
P_2
\leq
1
\label{eq:p2-kinematic-bounds}
\end{equation}
give corresponding constraints on the possible quartic structure. The
lower equality requires \(\rho_2\) to be completely mixed on the
two-fermion space,
\begin{equation}
\rho_2
=
\frac{1}{\binom{M}{2}}I_{d_2},
\end{equation}
while the upper equality would require a pure two-particle reduced state.

These are kinematic bounds on a general two-particle density matrix. They
need not be attainable by a pure \(N\)-fermion state that also satisfies
\(\Gamma=(N/M)I_M\). Determining the actual values of
\(P_2^{\min}\) and \(P_2^{\max}\), and identifying the states that attain
them, is therefore a new constrained representability problem within
\(\mathcal{S}_{M,N}^{\mathrm{flat}}\).

\subsection{Bounds and exactly solvable sectors}
\label{subsec:quartic-bounds}

Before addressing the general quartic extremal problem, it is useful to
examine the smallest particle-number sectors. These cases serve primarily
as consistency checks: the purity of the 2RDM is fixed by elementary
properties of the reduced state, so no further optimization is required.
They also identify where the second extremal problem first becomes
nontrivial.

For \(N=2\), the normalized 2RDM is the full pure state,
\begin{equation}
\rho_2 = |\psi\rangle\langle\psi|,
\end{equation}
and hence
\begin{equation}
P_2
\equiv
\operatorname{Tr}(\rho_2^2)
=
1.
\label{eq:N2-purity}
\end{equation}
Therefore, whenever the flat-1RDM manifold is nonempty,
\begin{equation}
P_{2,\mathrm{flat}}^{\min}
=
P_{2,\mathrm{flat}}^{\max}
=
1.
\end{equation}
The existence qualification is important. From the classification of
Sec.~\ref{sec:saturators}, a flat pure two-fermion state exists only for
even \(M\), with \(M=2\) corresponding to the trivial fully occupied
sector. Substitution into Eq.~\eqref{eq:q4-p2} gives
\begin{equation}
Q_4^2
=
1-\frac{8}{M^3}.
\label{eq:N2-q4}
\end{equation}
Thus, for \(N=2\), different flat-1RDM states may have different quartic
tensors, but their Frobenius magnitude cannot distinguish them.

The \(N=3\) sector is rigid for a different reason. For a pure
three-fermion state, the normalized one- and two-particle reduced states
are complementary reductions of the same pure state and therefore have
the same nonzero spectrum. If the 1RDM is flat,
\begin{equation}
\rho_1
=
\frac{\Gamma}{3}
=
\frac{1}{M}I_M,
\end{equation}
so the nonzero eigenvalues of \(\rho_2\) are also all equal to \(1/M\).
Consequently,
\begin{equation}
P_2
=
\operatorname{Tr}(\rho_2^2)
=
\frac{1}{M},
\qquad
\operatorname{Tr}(D^2)
=
\frac{9}{M}.
\label{eq:N3-purity}
\end{equation}
Whenever the flat manifold exists, the purity is therefore again fixed,
and Eq.~\eqref{eq:lambda-purity} gives
\begin{equation}
Q_4^2
=
\frac{9(2M+3)(M-3)}{2M^3}.
\label{eq:N3-q4}
\end{equation}
According to the existence classification above,
\(\mathcal{S}_{M,3}^{\mathrm{flat}}\) is nonempty for \(M=3\) and for
\(M\geq6\); the case \(M=3\) is again the trivial fully occupied sector.

The \(N=2\) and \(N=3\) results should therefore be viewed as boundary
checks rather than solutions of the general quartic optimization.
Starting at \(N=4\), the spectrum of the 2RDM is no longer determined by
the flatness of the 1RDM through these elementary arguments, and the
extremal problem becomes genuinely nontrivial.

We next consider what can be said in arbitrary particle-number sectors.
Since \(\rho_2\) is a normalized density matrix on the two-fermion space
of dimension
\begin{equation}
d_2=\binom{M}{2},
\end{equation}
its purity satisfies
\begin{equation}
P_2
=
\operatorname{Tr}(\rho_2^2)
\geq
\frac{1}{d_2}
=
\frac{1}{\binom{M}{2}}.
\label{eq:p2-lower}
\end{equation}
Equality holds if and only if the two-particle reduced state is maximally
mixed,
\begin{equation}
\rho_2
=
\frac{1}{\binom{M}{2}}I_{d_2},
\label{eq:maxmixed-2rdm}
\end{equation}
or equivalently,
\begin{equation}
D
=
\frac{\binom{N}{2}}{\binom{M}{2}}I_{d_2}.
\label{eq:maxmixed-D}
\end{equation}

Through Eq.~\eqref{eq:q4-p2}, this minimum-purity condition is precisely
the equality condition of the universal quartic lower bound established
in Corollary~\ref{cor:lambda-lower-bound}. The purity formulation is useful
here because it characterizes the structure required to reach that lower
edge: the normalized 2RDM must be maximally mixed.

Importantly, the condition in Eq.~\eqref{eq:maxmixed-2rdm} is already
strong enough to imply a flat 1RDM. Indeed, contracting the 2RDM gives
\begin{equation}
\sum_j D_{ij;kj}
=
(N-1)\Gamma_{ik},
\end{equation}
and Eq.~\eqref{eq:maxmixed-D} therefore yields
\begin{equation}
\Gamma
=
\frac{N}{M}I_M.
\end{equation}
Hence a pure state with a maximally mixed 2RDM automatically belongs to
\(\mathcal{S}_{M,N}^{\mathrm{flat}}\).

The lower-edge problem can therefore be formulated as an existence
question:
\begin{equation}
P_{2,\mathrm{flat}}^{\min}(M,N)
=
\frac{1}{\binom{M}{2}}
\quad\Longleftrightarrow\quad
\exists\,|\psi\rangle
\ \text{with}\
\rho_2=
\frac{1}{\binom{M}{2}}I_{d_2}.
\label{eq:lower-existence}
\end{equation}
If such a maximally mixed 2RDM is realized, the 1RDM is automatically
flat, as shown above. The remaining issue is therefore purely one of
representability: for which pairs \((M,N)\) can a pure \(N\)-fermion
state realize this equality condition? Such states are closely related
to fermionic states with maximally mixed few-particle reduced density
matrices \cite{GiorgadzeEtAl2025}. This gives a natural second-level
analogue of the existence problem encountered for first antiflatness.

The upper edge is more subtle. The general density-matrix bound
\(P_2\leq 1\) does not incorporate fermionic \(N\)-representability and
is therefore not expected to be sharp for \(N>2\). A stronger constraint
follows from a Hilbert--Schmidt estimate proved by Christiansen for the
connected, or truncated, fermionic two-body correlation operator
\cite{Christiansen2024}.

To express this result in the notation used here, recall that
\(\Lambda\) is the connected part of the 2RDM,
\begin{equation}
\Lambda
=
D-D^{\mathrm{Wick}}.
\end{equation}
Christiansen formulates the corresponding connected two-body operator on
the full tensor-product two-particle space, whereas throughout this work
we represent \(D\) and \(\Lambda\) in the normalized antisymmetric-pair
basis \(\{|ij\rangle_A:i<j\}\). With this normalization, the matrix
elements of the truncated two-body operator are twice those of
\(\Lambda\). Consequently, Christiansen's Hilbert--Schmidt estimate,
written directly in our convention, becomes
\begin{equation}
2\|\Lambda\|_F
\leq
\sqrt{
5N\,
\operatorname{Tr}
\left[
\Gamma(1-\Gamma)
\right]
}.
\label{eq:christiansen-lambda-bound}
\end{equation}
Here the Hilbert--Schmidt norm appearing in
Ref.~\cite{Christiansen2024} is identical to the Frobenius norm used
throughout this work for finite-dimensional operators.

The one-body quantity on the right-hand side is directly related to the
first antiflatness. From Eq.~\eqref{eq:f1_gamma},
\begin{align}
\operatorname{Tr}
\left[
\Gamma(1-\Gamma)
\right]
&=
\operatorname{Tr}\Gamma-\operatorname{Tr}(\Gamma^2)
\nonumber\\
&=
N-\operatorname{Tr}(\Gamma^2)
\nonumber\\
&=
\frac{F_1}{4}.
\label{eq:f1-trace-relation}
\end{align}
Substituting this into
Eq.~\eqref{eq:christiansen-lambda-bound} gives
\begin{equation}
2\|\Lambda\|_F
\leq
\frac{1}{2}\sqrt{5NF_1}.
\end{equation}
Since \(Q_4=\|\Lambda\|_F\), we obtain
\begin{equation}
Q_4
\leq
\frac{1}{4}\sqrt{5NF_1},
\end{equation}
or equivalently,
\begin{equation}
{
Q_4^2
\leq
\frac{5NF_1}{16}
}.
\label{eq:q4-upper-general}
\end{equation}

Thus, although \(F_1\) depends only on one-body information, it places an upper bound on the magnitude of the connected two-body
correlations. In particular, on the maximally antiflat manifold,
\begin{equation}
F_1
=
B_{M,N}
=
\frac{4N(M-N)}{M},
\end{equation}
and therefore
\begin{equation}
{
Q_4^2
\leq
\frac{5N^2(M-N)}{4M}
}.
\label{eq:q4-upper-flat}
\end{equation}

\section{Conclusion and outlook}

We have characterized the extremal structure of first fermionic
antiflatness for pure states at fixed particle number and used it as the
starting point for a hierarchy of higher-order fermionic correlations.
For a pure state of \(N\) fermions in \(M\) modes, particle-number
conservation reduces all quadratic information to the 1RDM \(\Gamma\),
and the first antiflatness becomes
\begin{equation}
F_1
=
4\left[
N-\operatorname{Tr}(\Gamma^2)
\right].
\end{equation}
This immediately gives the sharp bound
\begin{equation}
F_1
\leq
B_{M,N}
=
\frac{4N(M-N)}{M},
\end{equation}
with equality if and only if
\begin{equation}
\Gamma=\frac{N}{M}I_M.
\end{equation}
Thus maximal first antiflatness is exactly equivalent to complete
one-body flatness. The same flat natural-occupation spectrum also
maximizes the relative-entropy distance from the Gaussian set at fixed
\((M,N)\), whenever it is realizable by a pure state.

We then solved the corresponding pure-state existence problem. Writing
\(r=\min(N,M-N)\), a flat 1RDM is realizable precisely in the trivial
\(r=0\) sectors, in the \(r=2\) sectors with even \(M\), and in every
sector with \(r\geq3\). The remaining cases, \(r=1\) and \(r=2\) with
odd \(M\), are excluded by pure-state rank and pairing constraints. The
allowed sectors are constructive: block states cover the divisible and
half-filled cases, while cyclic-orbit states provide explicit saturators
in the remaining generic sectors. Consequently, the fixed-\(N\)
maximal-antiflatness problem is solved both at the level of its extremal
value and at the level of the pure states that realize it.

The covariance-level solution does not, however, determine the full
many-body correlation structure. Flat 1RDMs do not fix the connected four-point tensor, and the block and cyclic families provide explicit saturators on which this structure can be evaluated. In a fixed-particle-number
sector, we showed that the connected four-Majorana Wick-error tensor and
the connected 2RDM cumulant \(\Lambda\) encode the same quartic
information. This identifies
\begin{equation}
\Gamma
\longrightarrow
\Lambda
\longrightarrow
\text{higher connected cumulants}
\end{equation}
as a natural correlation hierarchy beyond covariance.

The first step of this hierarchy is already strongly constrained. For any
fixed-\(N\) state, the trace of the connected 2RDM is tied directly to
first antiflatness through
\begin{equation}
\operatorname{Tr}\Lambda=-\frac{F_1}{8}.
\end{equation}
Consequently, connected quartic correlations cannot vanish in any
nontrivial maximally antiflat state. Their minimum magnitude is associated
with a maximally mixed normalized 2RDM, so attainment of the lower edge
becomes a new pure-state representability problem. At fixed \((M,N)\),
we further showed that extremizing
\[
Q_4=\|\Lambda\|_F
\]
is equivalent to extremizing the purity of the normalized 2RDM. The
one-body extremal problem is therefore followed naturally by a distinct
two-body extremal problem.

A central result of the quartic analysis is that this second level is not
arbitrarily large either. By converting Christiansen's Hilbert--Schmidt
estimate for the truncated two-body operator
\cite{Christiansen2024} to the connected-2RDM convention used here, we
obtain
\begin{equation}
Q_4^2
\leq
\frac{5NF_1}{16}.
\label{eq:conclusion-q4-upper}
\end{equation}
For a maximally antiflat state this yields
\begin{equation}
Q_4^2
\leq
\frac{5N^2(M-N)}{4M}.
\label{eq:conclusion-q4-upper-flat}
\end{equation}
This provides a direct quantitative bridge between the covariance-level
resource measured by first antiflatness and the magnitude of connected
two-body correlations that remain beyond covariance. Although the bound
is inherited from a general fermionic estimate and need not be sharp on
the flat-1RDM manifold, it establishes a upper envelope for the
next level of the correlation hierarchy.

Taken together, these results sharpen the interpretation of fermionic
antiflatness. First antiflatness is not a complete measure of the
many-body departure from Gaussianity: states can agree exactly at the
covariance level while differing in their connected
four-point structure. Nevertheless, the one-body and two-body levels are
not independent. Fixed particle number enforces a nonzero quartic
contribution, while the antiflatness itself constrains how large that
contribution can become. The resulting picture is therefore not simply a
sequence of unrelated diagnostics, but a constrained hierarchy in which
higher-order correlations contain genuinely new information while
remaining linked to the lower-order data.

Several concrete problems remain open. At the lower edge, one must
determine for which \((M,N)\) a pure fermionic state can realize a
maximally mixed 2RDM and hence saturate the universal quartic minimum. At
the upper edge, the bound in
Eq.~\eqref{eq:conclusion-q4-upper-flat} is general rather than tailored
to the flat-1RDM manifold. Determining the sharp value of
\(Q_4^{\max}(M,N)\), identifying its equality conditions, and
establishing whether additional pure-state \(N\)-representability
constraints can strengthen the present upper bound are natural next
questions. It would also be useful to determine where the block and
cyclic saturators lie within these quartic extrema and whether either
family is extremal in distinguished sectors.

Beyond fourth order, the same program can be continued with connected
six-point and higher cumulants. This would produce a sequence of
constrained extremal problems in which the information invisible at one
correlation order becomes the object of study at the next. Such a
hierarchy may provide a systematic framework for separating efficiently
accessible covariance information from progressively higher-order
non-Gaussian resources in interacting fermionic systems, while retaining constraints between the different levels.

\section*{Author contributions}

R.G.: Conceptualization, methodology, formal analysis, writing initial draft. Formulated the original research question, derived the sharp fixed-particle-number antiflatness bound, and developed the construction and existence classification of maximally antiflat states.

A.H.: Methodology, formal analysis, writing -- review and editing.  Substantially rewrote and organized the manuscript, developed the higher-order correlation analysis in Sec.~IV, and derived the associated lower and upper bounds and the correlation-hierarchy formulation.

J.D.W.: Conceptualization, supervision and guidance, and manuscript writing.

\section*{Conflicts of interest}
There are no conflicts to declare.

\section*{Data availability}
The data and code supporting the findings of this study are available at
\url{https://github.com/Sciotaku/FAFresearch}.

\section*{Acknowledgments}
The authors thank members of the Whitfield group for helpful discussions. 
Our work is supported by the Army Research Office under award W911NF2410043 as well as the Department of Energy under award DE-SC0019374.  JDW holds concurrent appointments
at Dartmouth College and as an Amazon Visiting Academic.
This paper describes work performed at Dartmouth College
and is not associated with Amazon.

\nocite{Mazziotti1998,KutzelniggMukherjee1999,JuhaszMazziotti2006,
GigenaRossignoli2015,WalterDoranGrossChristandl2013,
AltunbulakKlyachko2008}
\bibliographystyle{apsrev4-2}
\bibliography{fermionic_antiflatness_refs}

\end{document}